\documentclass{article}

\usepackage{xspace}

\usepackage[utf8]{inputenc}
\usepackage[english]{babel}
\usepackage[T1]{fontenc}
\usepackage{lmodern}

\usepackage{amsmath, amsthm, amssymb}

\usepackage[affil-it]{authblk}

\usepackage{enumitem}
\usepackage{bussproofs}
\usepackage[colorlinks=true]{hyperref}
\usepackage[plain]{fancyref}
\usepackage[norefs,nocites]{refcheck}

\usepackage{tikz}
\usetikzlibrary{matrix}

\newcommand*{\fancyrefexlabelprefix}{ex}

\newcommand*{\frefexname}{\text{example}}
\newcommand*{\Frefexname}{\text{Example}}

\frefformat{vario}{\fancyrefexlabelprefix}{\frefexname\fancyrefdefaultspacing#1#3}
\Frefformat{vario}{\fancyrefexlabelprefix}{\Frefexname\fancyrefdefaultspacing#1#3}
\frefformat{plain}{\fancyrefexlabelprefix}{\frefexname\fancyrefdefaultspacing#1}
\Frefformat{plain}{\fancyrefexlabelprefix}{\Frefexname\fancyrefdefaultspacing#1}

\newcommand*{\fancyrefdeflabelprefix}{def}

\newcommand*{\frefdefname}{\text{definition}}
\newcommand*{\Frefdefname}{\text{Definition}}

\frefformat{vario}{\fancyrefdeflabelprefix}{\frefdefname\fancyrefdefaultspacing#1#3} 
\Frefformat{vario}{\fancyrefdeflabelprefix}{\Frefdefname\fancyrefdefaultspacing#1#3}
\frefformat{plain}{\fancyrefdeflabelprefix}{\frefdefname\fancyrefdefaultspacing#1}
\Frefformat{plain}{\fancyrefdeflabelprefix}{\Frefdefname\fancyrefdefaultspacing#1}

\newcommand*{\fancyrefthmlabelprefix}{thm}

\newcommand*{\frefthmname}{\text{theorem}}
\newcommand*{\Frefthmname}{\text{Theorem}}

\frefformat{vario}{\fancyrefthmlabelprefix}{\frefthmname\fancyrefdefaultspacing#1#3}
\Frefformat{vario}{\fancyrefthmlabelprefix}{\Frefthmname\fancyrefdefaultspacing#1#3}
\frefformat{plain}{\fancyrefthmlabelprefix}{\frefthmname\fancyrefdefaultspacing#1}
\Frefformat{plain}{\fancyrefthmlabelprefix}{\Frefthmname\fancyrefdefaultspacing#1}

\newcommand*{\fancyrefremlabelprefix}{rem}

\newcommand*{\frefremname}{\text{remark}}
\newcommand*{\Frefremname}{\text{Remark}}

\frefformat{vario}{\fancyrefremlabelprefix}{\frefremname\fancyrefdefaultspacing#1#3}
\Frefformat{vario}{\fancyrefremlabelprefix}{\Frefremname\fancyrefdefaultspacing#1#3}
\frefformat{plain}{\fancyrefremlabelprefix}{\frefremname\fancyrefdefaultspacing#1}
\Frefformat{plain}{\fancyrefremlabelprefix}{\Frefremname\fancyrefdefaultspacing#1}

\newcommand*{\fancyreflemlabelprefix}{lem}

\newcommand*{\freflemname}{\text{lemma}}
\newcommand*{\Freflemname}{\text{Lemma}}

\frefformat{vario}{\fancyreflemlabelprefix}{\freflemname\fancyrefdefaultspacing#1#3}
\Frefformat{vario}{\fancyreflemlabelprefix}{\Freflemname\fancyrefdefaultspacing#1#3}
\frefformat{plain}{\fancyreflemlabelprefix}{\freflemname\fancyrefdefaultspacing#1}
\Frefformat{plain}{\fancyreflemlabelprefix}{\Freflemname\fancyrefdefaultspacing#1}

\newcommand*{\fancyrefsubseclabelprefix}{subsec}

\newcommand*{\frefsubsecname}{\text{subsection}}
\newcommand*{\Frefsubsecname}{\text{Subsection}}

\frefformat{vario}{\fancyrefdeflabelprefix}{\frefsubsecname\fancyrefdefaultspacing#1#3}
\Frefformat{vario}{\fancyrefsubseclabelprefix}{\Frefsubsecname\fancyrefdefaultspacing#1#3}
\frefformat{plain}{\fancyrefsubseclabelprefix}{\frefsubsecname\fancyrefdefaultspacing#1}
\Frefformat{plain}{\fancyrefsubseclabelprefix}{\Frefsubsecname\fancyrefdefaultspacing#1}

\newcommand*{\fancyrefcorlabelprefix}{cor}

\newcommand*{\frefcorname}{\text{corollary}}
\newcommand*{\Frefcorname}{\text{Corollary}}

\frefformat{vario}{\fancyrefcorlabelprefix}{\frefcorname\fancyrefdefaultspacing#1#3}
\Frefformat{vario}{\fancyrefcorlabelprefix}{\Frefcorname\fancyrefdefaultspacing#1#3}
\frefformat{plain}{\fancyrefcorlabelprefix}{\frefcorname\fancyrefdefaultspacing#1}
\Frefformat{plain}{\fancyrefcorlabelprefix}{\Frefcorname\fancyrefdefaultspacing#1}

\newcommand*{\fancyrefsubsubseclabelprefix}{subsubsec}

\newcommand*{\frefsubsubsecname}{\text{subsubsection}}
\newcommand*{\Frefsubsubsecname}{\text{Subsubsection}}

\frefformat{vario}{\fancyrefdeflabelprefix}{\frefsubsubsecname\fancyrefdefaultspacing#1#3}
\Frefformat{vario}{\fancyrefsubsubseclabelprefix}{\Frefsubsubsecname\fancyrefdefaultspacing#1#3}
\frefformat{plain}{\fancyrefsubsubseclabelprefix}{\frefsubsubsecname\fancyrefdefaultspacing#1}
\Frefformat{plain}{\fancyrefsubsubseclabelprefix}{\Frefsubsubsecname\fancyrefdefaultspacing#1}

\newcommand*{\fancyrefalgolabelprefix}{algo}

\newcommand*{\frefalgoname}{\text{algorithm}}
\newcommand*{\Frefalgoname}{\text{Algorithm}}

\frefformat{vario}{\fancyrefalgolabelprefix}{\frefalgoname\fancyrefdefaultspacing#1#3}
\Frefformat{vario}{\fancyrefalgolabelprefix}{\Frefalgoname\fancyrefdefaultspacing#1#3}
\frefformat{plain}{\fancyrefalgolabelprefix}{\frefalgoname\fancyrefdefaultspacing#1}
\Frefformat{plain}{\fancyrefalgolabelprefix}{\Frefalgoname\fancyrefdefaultspacing#1}

\newcommand*{\fancyrefproplabelprefix}{prop}

\newcommand*{\frefpropname}{\text{proposition}}
\newcommand*{\Frefpropname}{\text{Proposition}}

\frefformat{vario}{\fancyrefproplabelprefix}{\frefpropname\fancyrefdefaultspacing#1#3}
\Frefformat{vario}{\fancyrefproplabelprefix}{\Frefpropname\fancyrefdefaultspacing#1#3}
\frefformat{plain}{\fancyrefproplabelprefix}{\frefpropname\fancyrefdefaultspacing#1}
\Frefformat{plain}{\fancyrefproplabelprefix}{\Frefpropname\fancyrefdefaultspacing#1}

\newtheorem{counter}{Counter}

\theoremstyle{definition}

\newtheorem{definition}[counter]{Definition}

\theoremstyle{plain}
\newtheorem{lemma}[counter]{Lemma}
\newtheorem{theorem}[counter]{Theorem}

\theoremstyle{remark}

\let\oldtextbf\textbf
\renewcommand{\textbf}[1]{\oldtextbf{\boldmath #1}}

\newcommand\Prop{\ensuremath{\mathsf{Prop}}\xspace}

\newcommand\RatS{\ensuremath{\mathsf{S}}\xspace}
\newcommand\Nat{\ensuremath{\mathbb{N}}\xspace}
\newcommand\Tm{\ensuremath{\mathsf{Tm}}\xspace}

\newcommand\LP{\ensuremath{\mathcal{L}_{\mathsf{P}}}\xspace}

\newcommand\LJ{\ensuremath{\mathcal{L_{\mathsf{J}}}}\xspace}

\newcommand\axP{\ensuremath{(\mathsf{P}})\xspace}
\newcommand\axJ{\ensuremath{(\mathsf{J}})\xspace}
\newcommand\axPlus{\ensuremath{(\mathsf{+}})\xspace}

\newcommand\ANE{(\ensuremath{\mathsf{AN!}})\xspace}
\newcommand\PI{(\ensuremath{\mathsf{PI}})\xspace}
\newcommand\WE{(\ensuremath{\mathsf{WE}})\xspace}
\newcommand\LE{(\ensuremath{\mathsf{LE}})\xspace}
\newcommand\DIS{(\ensuremath{\mathsf{DIS}})\xspace}
\newcommand\UN{(\ensuremath{\mathsf{UN}})\xspace}

\newcommand\MP{(\ensuremath{\mathsf{MP}})\xspace}
\newcommand\CE{\ensuremath{(\mathsf{CE}})\xspace}
\newcommand\ST{\ensuremath{(\mathsf{ST}})\xspace}

\newcommand\logic{\ensuremath{\mathsf{L}}\xspace}

\newcommand\PJCS{\ensuremath{{\PJ}_{\CS}}\xspace}
\newcommand\JCS{\ensuremath{\mathsf{J}_{\mathsf{CS}}}\xspace}
\newcommand\PPJ{\ensuremath{\mathsf{PPJ}}\xspace}
\newcommand\PJ{\ensuremath{\mathsf{PJ}}\xspace}
\newcommand\J{\ensuremath{\mathsf{J}}\xspace}
\newcommand\LPPTwo{\ensuremath{\mathsf{LPP_2}}\xspace}

\newcommand\PJCSMeas{\ensuremath{\mathsf{{\PJ}_{\CS,Meas}}}\xspace}

\newcommand\true{\ensuremath{\mathsf{T}}\xspace}
\newcommand\false{\ensuremath{\mathsf{F}}\xspace}

\newcommand\CS{\ensuremath{\mathsf{CS}}\xspace}
\newcommand\powerset{\ensuremath{\mathcal{P}}\xspace}
\newcommand\subf{\ensuremath{\mathsf{subf}}\xspace}

\newcommand\system{\ensuremath{\mathcal{S}}\xspace}
\newcommand\myvec[1]{\ensuremath{\boldsymbol{#1}}\xspace}
\newcommand\op{\ensuremath{\mathsf{\Diamond}}\xspace}

\title{The Complexity of Non-Iterated
Probabilistic Justification Logic}
\author{Ioannis Kokkinis}

\affil{Institute of Computer Science, University of Bern, Switzerland}

\begin{document}

\maketitle

\begin{abstract}
The logic \PJ is a probabilistic logic
defined by adding (non-iterated) probability
operators to the basic
justification logic \J.
In this paper we establish upper and lower bounds for the
complexity of the derivability problem in the logic
\PJ. The main result of the
paper is that
the complexity of the derivability problem
in \PJ remains the same as the complexity of the 
derivability problem in the underlying
logic \J, which is $\Pi_2^p$-complete.
This implies hat the probability operators do not increase the 
complexity of the
logic, although they arguably enrich the expressiveness
of the language.
\end{abstract}


\section{Introduction}

Traditional modal epistemic logic uses formulas
of the form $\Box \alpha$ to express that an agent believes 
$\alpha$. The language of justification logic~\cite{ArtFit11SEP} 
`unfolds' the $\Box$-modality into a family of so-called 
\emph{justification terms}, which
are used to represent evidence for the agent's belief.
Hence, instead of $\Box \alpha$,
justification logic includes formulas of the form $t : \alpha$ 
meaning 
\begin{center}
the agent believes $\alpha$ for reason $t$.
\end{center}
Artemov~\cite{Art95TR,Art01BSL} developed the first justification 
logic, the Logic of Proofs, to provide intuitionistic logic with 
a classical provability semantics. There, justification terms 
represent formal proofs in Peano Arithmetic. However, terms may 
also represent
informal justifications. For instance, our belief in
$\alpha$ may be
justified by direct observation of $\alpha$ or by learning that a 
friend heard about $\alpha$.
This general reading of justification led to a big variety of 
epistemic justification
logics for many different applications 
\cite{BucKuzStu11JANCL,BucKuzStu11WoLLIC,KuzStu12AiML}.
In~\cite{komaogst,koogst} we extended justification logic with probability operators in order to accommodate the idea that 
\begin{center}
different kinds of evidence for $\alpha$
lead to different degrees of belief in $\alpha$.
\end{center}
For example it could be the case that the agent learns
$\alpha$ from some unreliable source (e.g. from
some friend of his) or that the agent reads about
$\alpha$ in some reliable newspaper. 
In both cases the agent has a justification
for $\alpha$: in the first case he has the statement
of his friend and in the second case the text of
the newspaper. However, it is natural that the agent
does not want to put the same credence in both sources
of information. This differentiation in credulity cannot
be expressed in classical justification logic. So,
the main contribution of justification logics with
probability operators (probabilistic justification
logics~\cite{komaogst,koogst}) is the ability to compare different sources of
information.
Uncertain reasoning in justification logic has also
been studied in 
\cite{Mil14APALnonote,2014arXiv1407.4647G,FanL15}.
See \cite{komaogst,koogst} for an extended comparison
between our approach and the ones from
\cite{Mil14APALnonote,2014arXiv1407.4647G,FanL15}.

Probabilistic logics are logics than can be
used to model uncertain reasoning. Although the
idea of probabilistic logic was first proposed by Leibnitz,
the modern development of this topic started only in the 1970s 
and 1980s in the papers of
H. Jerome Keisler~\cite{kei77} and Nils Nilsson~\cite{nil86}. 
Following Nilsson's research, Fagin, Halpern and
Meggido~\cite{fahame90} introduced a logic with arithmetical 
operations built into the syntax
so that Boolean combinations of linear inequalities of 
probabilities of formulas
can be expressed. The probabilistic logic of 
\cite{fahame90} can be considered as a probabilistic
logic with classical base. The derivability problem in
this logic is proved to be $coNP$-complete, the
same as that of classical propositional logic.
Following the lines of~\cite{fahame90},
Ognjanovi\'{c}, Ra\v{s}kovi\'{c} and Markovi\'{c}~\cite{ograma09}
defined the logic \LPPTwo, which is also a probabilistic
logic with classical base. The logic \LPPTwo makes
use of an infinitary rule which makes the proof of strong
completeness possible (as opposed to the finitary system of 
\cite{fahame90} which is only simply complete).
The \LPPTwo-derivability problem is again $coNP$-complete. 

Following the lines of 
\cite{ograma09} the logic \PJ was defined
in \cite{komaogst}.
\PJ is a probabilistic logic
defined over the basic justification logic
\J.\footnote{\J stands for justification, whereas \PJ
stands for probabilistic justification.}
The language of \PJ contains formulas of the form
$P_{\geq s} \alpha$ meaning
\begin{center}
the probability of truthfulness
of the justification formula $\alpha$
is at least $s$.
\end{center}
So, in the logic \PJ, statements like ``evidence
$t$ serves as a justification for $\alpha$ with
probability at least 30\%'' can be expressed.
\PJ
does not allow iterations of the probability operator.
In~\cite{koogst} we study an extension of \PJ, the logic
\PPJ,\footnote{the two P's stand for iterations of the
probability operator.}
where iterations of the probability operator
as well as justification operators over
probability operators are allowed.

The results of
\cite{Kuz00CSLnonote,Mil07APAL,BusKuz12APALnonote,ach15} showed
that, under some reasonable assumptions, the derivability problem for the justification
logic \J is $\Pi^p_2$-complete, i.e.~it is complete
in the second level of the polynomial hierarchy. In this
paper we show that under the same assumptions
the derivability problem for the
probabilistic
justification logic \PJ remains in
the class $\Pi^p_2$-complete. 
We achieve this, by showing that the satisfiability
problem for the logic \PJ, which is dual
to the derivability problem, belongs
to the class $\Sigma^p_2$-complete.
The methods we use are adaptations from
\cite{fahame90} and \cite{Kuz00CSLnonote}.
As it is the case
in \cite{ograma09} and \cite{fahame90} we also make use
of some well known results from the theory of linear
programming.
The main result of the paper is that the probability
operators do not increase the complexity of the
logic, although they arguably
enrich the expressiveness of the logical framework.

The rest of the paper is organized as follows.
In section \ref{sec:JandPJ} we briefly recall the justification
logic \J and the probabilistic justification logic \PJ.
In section \ref{sec:smp} we establish a small model
theorem for \PJ. In section \ref{sec:compl} we present
an algorithm that decides the satisfiability problem
for the logic \PJ and evaluate its complexity.
We close the paper in section \ref{sec:concl} with
some final observations.

An earlier version of the present paper is available in
arXiv~\cite{kok15}.

\section{The logics \J and \PJ}

\label{sec:JandPJ}

In this section we briefly recall the basic
justification logic \J\cite{ArtFit11SEP} and the probabilistic 
justification logic \PJ~\cite{komaogst}. 

\emph{Justification terms} are built according to the 
following grammar:
\[
t :: = c ~|~ x ~|~ (t \cdot t) ~|~ (t+t) ~|~!t
\]
where $c$ is a constant and $x$ is a variable. 
$\Tm$ denotes the set of all terms.
For any term $t$ 
and any non-negative integer~$n$ we define:
\[
!^0 t := t
\qquad\text{and}\qquad
!^{n+1}t := {!} ~ ({!^n} t)
\]
Terms are used to provide justifications
for formulas. Constants are used as justifications for
axioms, whereas variables are used as
justifications for arbitrary formulas.
The operator $\cdot$ can be used by the agents
to apply modus ponens (see axiom \axJ in 
Figure~\ref{fig:Jaxioms}), the operator $+$ is used
for concatenation of proofs (see axiom \axPlus in 
Figure~\ref{fig:Jaxioms}) and the operator $!$ is
used for stating positive introspection (see rule \ANE in 
Figure~\ref{fig:JCSFig}). That is, if the agent
has a justification $c$ for $\alpha$ then he has a
justification $!c$ for the fact that $c$ 
is a justification for $\alpha$ and so on.
 
Let \Prop denote a countable set of atomic propositions.
Formulas of the language \LJ (justification
formulas) are 
built according to the following grammar:
\[
\alpha :: = p ~|~ \lnot \alpha ~|~ \alpha \land \alpha 
~|~ t: \alpha
\]
where $t \in \Tm$ and $p \in \Prop$.
Any formula of the form
$t: \alpha$ for $t \in \Tm$ and $\alpha \in \LJ$
will be called a \emph{justification assertion}.
We will use the letter $p$ possibly primed or with subscripts
to represent an element of \Prop  and lower-case Greek letters 
like $\alpha, \beta, \gamma, \ldots$ for \LJ-formulas.
In Figure~\ref{fig:Jaxioms} we present the axioms 
schemes of the logic \J. 

\begin{figure}[ht]
\centering
\renewcommand*{\arraystretch}{1.25}
\begin{tabular}{|c l|}
\hline
\axP & finitely many axiom schemes for classical\\
& propositional
logic in the language of \LJ \\
\axJ & $\vdash u : (\alpha \to \beta) \to
( v :\alpha \to u \cdot v : \beta ) $\\
\axPlus & $\vdash \big ( u: \alpha \lor v : \alpha \big ) \to  u+v: \alpha$\\
\hline
\end{tabular}
\caption{\label{fig:Jaxioms} Axioms Schemes of \J}
\end{figure}

In order to build justifications for 
arbitrary formulas in the logic
\J we need to start by some justifications for the
axioms. That is why we need the notion of
a constant specification.
A \emph{constant specification} is any set
\CS that satisfies the following condition:
\begin{align*}
\CS \subseteq \{(c,\alpha)~|~ & c \text{ is a constant and
$\alpha$ is an instance }\\
& \text{of some axiom scheme of the logic \J } \}
\end{align*}

A constant specification \CS will be called:

\begin{description}[topsep=0em]
\item[axiomatically appropriate:]
if for every instance of a \J-axiom, $\alpha$, there exists some
constant $c$ such that $(c, \alpha) \in \CS$, i.e.
every instance of a \J-axiom scheme is justified by at least one
constant.
\item[schematic:]
if for every constant $c$ the set
\[
\big \{ \alpha ~ \big |~ (c,\alpha) \in \CS \big \}
\]
consists of all instances of several (possibly zero)
axiom schemes, i.e. if every constant specifies
certain axiom schemes and only them.
\item[decidable:] if the set \CS is decidable. In this
paper when we refer to a decidable \CS, we will always
imply that \CS is decidable in \emph{polynomial time}.
\item[finite:] if \CS is a finite set.
\item[almost schematic:]
if $\CS = \CS_1 \cup \CS_2$ where $\CS_1 \cap \CS_2 = \emptyset$,
$\CS_1$ is a schematic constant specification
and $\CS_2$ is a finite constant specification.
\item[total:] if for every constant $c$ and
every instance $\alpha$
of a \J-axiom scheme, $(c,\alpha) \in \CS$.
\end{description}

Let \CS be any constant specification.
The deductive system \JCS is presented in
Figure~\ref{fig:JCSFig}. 
\begin{figure}[h]
\centering
\renewcommand*{\arraystretch}{1.25}
\begin{tabular}{|c l|}
\hline 
& axioms schemes of \J \\
& ~~~~~+ \\
\ANE & $\vdash {!^{n+1}} c : {!^n} c : \cdots : {!c} : c : \alpha$,
where $(c, \alpha) \in \CS$ and $n \in \Nat$\\
\MP & \text{if }$T \vdash \alpha \text{ and } T \vdash \alpha \to \beta \text{ then } T \vdash \beta$\\
\hline
\end{tabular}
\caption{\label{fig:JCSFig} System \JCS}
\end{figure}

As usual $T \vdash_{\logic} \alpha$ means that the formula
$\alpha$ is provable from the set of formulas $T$ using the
rules and axioms of the logic \logic. When \logic is
clear from the context it will be omitted.

Now we present the semantics for the logic \J.
The models for a \JCS are the so called
\JCS-evaluations (see Definition~\ref{def:BasicEval}).
We use \true to represent the truth value ``true'' and
\false to represent the truth value ``false''.
Let $\powerset(W)$ denote the powerset of
the set $W$.

\begin{definition}[\JCS-Evaluation]
\label{def:BasicEval}
Let \CS be any constant specification.
A \JCS-evaluation is a function $*$ 
such that $* : \Prop \to \{\true, \false\}$
and $*: \Tm \to \powerset(\LJ)$ and for
$u, v \in \Tm$, for a constant $c$ and $\alpha,\beta
\in \LJ$ we
have:
\begin{enumerate}[topsep = 0.25em, label = (\arabic*)]
\item
$ \big ( \alpha \to \beta \in u^*  \text{ and }
\alpha  \in v^* \big ) \Longrightarrow
\beta \in (u \cdot v)^*$
\item
$u^* \cup v^* \subseteq (u + v)^*$
\item
if $(c, \alpha) \in \CS$ then for all 
$n \in \Nat$ we have\footnote{We agree to the convention that the formula
${!^{n-1}} c : {!^{n-2}} c : \cdots : {!c} : c : \alpha$
represents the formula $\alpha$ for $n=0$.}: 
\[
{!^{n-1} c} : {!^{n-2}} c: \cdots 
: !c : c : \alpha \in 
({!^n} c)^*
\]
\end{enumerate}
We will usually write $t^*$ and $p^*$ instead of
$*(t)$ and $*(p)$ respectively.
\end{definition}

Now we will define the binary relation $\Vdash$.

\begin{definition}[Truth under a \JCS-Evaluation]
We define what it means for an \LJ-formula
to hold under a
\JCS-evaluation $*$ inductively as follows:
\begin{align*}
* \Vdash p & \Longleftrightarrow p^* = \true\\
* \Vdash \lnot \alpha & \Longleftrightarrow * \not\Vdash \alpha\\
* \Vdash \alpha \land \beta & \Longleftrightarrow
\big (* \Vdash \alpha \text{ and } * \Vdash \beta \big)\\
* \Vdash t:\alpha & \Longleftrightarrow \alpha \in t^*
\end{align*}
\end{definition}

We have the following theorem.

\begin{theorem}[Completeness of \J\cite{Art12SL,KuzStu12AiML}]
Let \CS be any constant specification.
Let $\alpha \in \LJ$. Then we have:
\[
\vdash_{\JCS} \alpha \quad \Longleftrightarrow \quad \Vdash_{\CS} \alpha .
\]
where $\Vdash_{\CS} \alpha$ means that $\alpha$ holds
under any \JCS-evaluation.
\end{theorem}

Let \RatS be the set of all rational numbers from the
interval $[0,1]$.
The formulas of the language \LP (the so called
probabilistic formulas) are built according to 
the following grammar:
\[
A :: = P_{\geq s} \alpha ~|~ \lnot A ~|~ A \land A
\]
where $s \in \RatS$, and $\alpha \in \LJ$. We use capital Latin 
letters like $A, B, C, \ldots$
for \LP-formulas.
We employ the standard abbreviations for classical connectives. 
Additionally, we set:
\begin{align*}
P_{< s} \alpha & \equiv \lnot P_{\geq s} \alpha \qquad
&P_{\leq s} \alpha & \equiv  P_{\geq 1 - s} \lnot \alpha \\
P_{> s} \alpha  & \equiv \lnot P_{\leq s} \alpha 
&P_{= s} \alpha & \equiv P_{\geq s} \alpha \land P_{\leq s}\alpha
\end{align*}

The axioms schemes of \PJ are presented in
Figure~\ref{fig:PJAxioms}.
\begin{figure}[ht]
\centering
\renewcommand*{\arraystretch}{1.25}
\begin{tabular}{|c l|}
\hline 
\axP & finitely many axiom schemes for
classical\\
& propositional logic in the language of \LP \\
\PI & $\vdash P_{\geq 0} \alpha$\\
\WE & $\vdash P_{\leq r} \alpha \to P_{< s} \alpha$, where $s > r$\\
\LE & $\vdash P_{< s} \alpha \to P_{\leq s} \alpha$\\
\DIS & $\vdash  P_{\geq r} \alpha \land P_{\geq s} \beta \land P_{\geq 1} \lnot (\alpha \land \beta) \to P_{\geq \min(1, r+s)} (\alpha \lor \beta)$\\
\UN & $\vdash P_{\leq r} \alpha \land P_{< s} \beta \to P_{<r+s} (\alpha \lor \beta)$, where $r+s \leq 1$\\
\hline
\end{tabular}
\caption{\label{fig:PJAxioms} Axioms Schemes of \PJ}
\end{figure}
For any constant specification \CS the deductive
system \PJCS is presented in
Figure~\ref{fig:PJCSFig}.
Definitions~\ref{def:algebra}--\ref{def:PJmodel} describe
the semantics for the logic \PJ.

\begin{figure}[ht]
\centering
\renewcommand*{\arraystretch}{1.25}
\begin{tabular}{|c l|}
\hline 
& axiom schemes of \PJ\\
& ~~~~+\\
\MP & if $T \vdash A$ and  $ T \vdash A \to B$ 
then $T \vdash B$\\
\CE & if $\vdash_{\JCS} \alpha$ 
then $\vdash_{\PJCS} P_{\geq 1 } \alpha$\\
\ST & if $T \vdash A \to P_{\geq s - \frac{1}{k} } \alpha$
for every integer $k \geq \frac{1}{s}$ and $s > 0$\\
& then $T \vdash A \to P_{\geq s } \alpha$\\
\hline
\end{tabular}
\caption{\label{fig:PJCSFig} System \PJCS}
\end{figure}

\begin{definition}[Algebra over a set]
\label{def:algebra}
Let $W$ be a non-empty set and let $H$ be a non-empty subset
of $\powerset(W)$. $H$ will be called an
\emph{algebra over $W$} iff the following hold:
\begin{itemize}[topsep = 0em]
\item
$W \in H$
\item
$U,V \in H \Longrightarrow U \cup V \in H$
\item
$U \in H \Longrightarrow W \setminus U \in H$
\end{itemize}
\end{definition}

\begin{definition}[Finitely Additive Measure]
\label{def:FinAdd}
Let $H$ be an algebra over $W$ and $\mu : H \to [0,1]$.
We call $\mu$  a
\emph{finitely additive measure} iff the following hold:
\begin{enumerate}[topsep = 0em, label = (\arabic*)]
\item
$\mu(W) = 1$
\item
for all $U, V \in H$:
\[
U \cap V = \emptyset \Longrightarrow
\mu(U \cup V) = \mu(U) + \mu(V)
\]
\end{enumerate}
\end{definition}

\begin{definition}[$\PJCS$-Model]
\label{def:PJmodel}
Let \CS be any constant specification.
A \emph{\PJCS-model}, or simply a model, is a structure 
$M = \langle W, H, \allowbreak \mu, * \rangle$ where:
\begin{itemize}[topsep = 0em]
\item
$W$ is a non-empty set of objects called worlds.
\item
$H$ is an algebra over $W$.
\item
$\mu : H \to [0,1]$ is a finitely additive measure.
\item
$*$ is a function from $W$ to the set of all
\JCS-evaluations, i.e.~$*(w)$ is a \JCS-evaluation
for each world $w \in W$. We will usually write $*_w$
instead of $*(w)$.
\end{itemize}
\end{definition}

\begin{definition}[Measurable model]
\label{def:measModel}
Let $M = \langle W, H, \allowbreak \mu, * \rangle$ be a model and $\alpha \in \LJ$. We define the 
following set:
\[
[\alpha]_M =
\{ w \in W ~|~ *_w \Vdash \alpha\}
\]
We will omit the subscript $M$, i.e.~we will simply write
$[\alpha]$, if $M$ is clear from the context.
A \PJCS-model $M = \langle W, H, \mu, *\rangle$ is 
\emph{measurable} iff $[\alpha]_M \in H$ for every
$\alpha \in \LJ$. The class of measurable \PJCS-models
will be denoted by \PJCSMeas.
\end{definition}

\begin{definition}[Truth in a $\PJCSMeas$-model]
\label{PJCSTruth}
Let \CS be any constant specification.
Let $M = \langle W, H, \mu, * \rangle$ be a \PJCSMeas-model.
We define what it means for an \LP-formula to hold
in $M$ inductively as follows\footnote{Observe that the
satisfiability relation of a \JCS-evaluation is
represented with $\Vdash$ whereas the satisfiability
relation of a model is represented with $\models$.}:
\begin{align*}
M \models P_{\geq s} \alpha & \Longleftrightarrow
\mu([\alpha]_M) \geq s\\
M \models \lnot A & \Longleftrightarrow M \not\models A\\
M \models A \land B & \Longleftrightarrow \big ( M \models A 
\text{ and } M \models B \big )
\end{align*}

\end{definition}

In the sequel we may refer to \PJCSMeas-models simply
as models if there is no danger for confusion.
We have the following theorem.

\begin{theorem}[Strong Completeness for
\PJ \cite{komaogst}]
\label{thm:PJsoundCompl}
Any \PJCS is sound and 
strongly complete with respect to \PJCSMeas-models, i.e.
for any $T \subseteq \LP$ and any $A \in \LP$:
\[
T \vdash_{\PJCS} A \Longleftrightarrow T \models_{\PJCS} A 
\]
\end{theorem}

Let \CS be any constant specification. 
A formula $A \in \LP$ is satisfied in
$M \in \PJCSMeas$ iff $M \models A$. $A$ will be called
\PJCSMeas-satisfiable or simply satisfiable 
if there is a 
\PJCSMeas-model that satisfies $A$.
We define the \PJCSMeas-satisfiability problem to be
the decision problem defined as follows:

\begin{center}
``For a given $A \in \LP$ and a given \CS is
$A$ \PJCSMeas-satisfiable?''
\end{center}

A formula $\alpha \in \LJ$ is satisfied in
a \JCS-evaluation $*$ iff $* \Vdash \alpha$. $\alpha$ 
will be called
\JCS-satisfiable or simply satisfiable 
if there is some \JCS-evaluation $*$ that satisfies $\alpha$.
We define the \JCS-satisfiability problem to be
the decision problem defined as follows:

\begin{center}
``For a given $\alpha \in \LJ$ and a given \CS is
$\alpha$ \JCS-satisfiable?''
\end{center}

\section{Small Model Property}

\label{sec:smp}

The goal of this section is to prove a small model
property for the logic \PJ.
The small model
property will be the most
important tool for establishing the upper bound for 
the complexity of \PJ.

\begin{definition}[Subformulas]
The set $\subf(\cdot)$ is defined recursively as follows:\\
\textbf{For \LJ-formulas:}
\begin{itemize}[topsep=0em]
\item
$\subf(p) := \{ p \}$
\item
$\subf(t : \alpha) := \{ t : \alpha \} \cup \subf (\alpha)$
\item
$\subf(\lnot \alpha) := \{ \lnot \alpha \} \cup \subf (\alpha)$
\item
$\subf(\alpha \land \beta) := \{ \alpha \land \beta\}\cup
\subf (\alpha) \cup \subf (\beta)$
\end{itemize}
\textbf{For \LP-formulas:}
\begin{itemize}[topsep=0em]
\item
$\subf(P_{\geq s} \alpha) := \{ P_{\geq s} \alpha \} \cup \subf 
(\alpha)$
\item
$\subf(\lnot A) := \{ \lnot A \} \cup \subf (A)$
\item
$\subf(A \land B) := \{ A \land B\} \cup  \subf (A) \cup \subf
(B)$
\end{itemize}
Observe that for $A \in \LP$ we have that
$\subf(A) \subseteq \LP \cup \LJ$.

\end{definition}

\begin{definition}[Atoms]
Let $A$ be an $\LP$- or an $\LJ$-formula.
Let $X$ be the set that
contains all the atomic propositions and the
justification assertions from the set
$\subf(A)$. An atom of $A$ is any  formula of the following
form:
\begin{equation}
\label{decid2}
\bigwedge_{B \in X} \pm B
\end{equation}
where $\pm B$ denotes either $B$ or $\lnot B$.
We will use the lowercase
Latin letter $a$ for atoms, possibly with subscripts.
\end{definition}

Let $A$ be an \LP- or an  \LJ-formula.
Assume that $A$ is either of the form
$\bigwedge_{i} B_i$ or of the form
$\bigvee_{i} B_i$. Then
$C \in A$ means that for some $i$, $B_i = C$.

\begin{definition}[Sizes]
The size function $|\cdot|$ is defined as follows:\\
\textbf{For \LP-formulas:} (recursively)
\begin{itemize}[topsep=0em]
\item
${\mid} P_{\geq s} \alpha {\mid} := 2$
\item
${\mid}\lnot A{\mid} := 1+{\mid}A{\mid}$
\item
${\mid}A \land B {\mid} := {\mid}A{\mid} + 1 + {\mid}B{\mid}$
\end{itemize}
\textbf{For sets:}\\
Let $W$ be a set. 
$|W|$ is the cardinal number of $W$.\\
\textbf{For non-negative integers:}\\
Let $r$ be an non-negative integer. We define the size
of $r$ to be equal to
the length of $r$ written in binary, i.e.:
\[
|r| := 
\begin{cases}
1 & , r = 0\\
\lfloor \log_2 (r) + 1 \rfloor & , r \geq 1
\end{cases}
\]
where $\lfloor \cdot \rfloor$ is the function that
returns the greatest integer that is less than or
equal to its argument.\\
\textbf{For non-negative rational numbers:}\\
Let $r = \frac{s_1}{s_2}$, where $s_1$ and
$s_2$ are relatively prime non-negative
integers with $s_2 \neq 0$, 
be a non-negative rational number. We define:
\[
|r| :=  |s_1| + |s_2|
\]
\end{definition}

Let $A \in \LP$ we define:
\[
||A|| := \max \big  \{ |s| ~ \big |~ P_{\geq s } \alpha \in \subf(A) 
\big \}
\]

Lemma~\ref{lem:prop} was originally proved in
\cite{ograma09} for the logic \LPPTwo.
The proof for the logic \PJ is given in \cite{komaogst}.

\begin{lemma} 
\label{lem:prop}
For any constant specification \CS, we have:
\[
\vdash_{\JCS} \alpha \leftrightarrow \beta
\Longleftrightarrow
~\vdash_{\PJCS} P_{\geq s} \alpha \leftrightarrow
P_{\geq s} \beta
\]
\end{lemma}

A proof for Theorem~\ref{thm:eqThm} can be found in \cite[p. 145]{chvatal83}.

\begin{theorem}
\label{thm:eqThm}
Let \system be a system of $r$ linear equalities. Assume that the 
vector\footnote{We will always
use bold font for vectors.} $\myvec{x}$
is a solution of \system such that all of
$\myvec{x}$'s entries are non-negative.
Then there is a vector $\myvec{x^*}$ such that:
\begin{enumerate}[label=(\arabic*), topsep = 0em]
\item
\label{enum:eqThmSol}
$\myvec{x^*}$ is a solution of \system.
\item
all the entries of $\myvec{x^*}$ are non-negative.
\item
\label{enum:eqThmRpos}
at most $r$ entries of $\myvec{x^*}$ are positive.
\end{enumerate}
\end{theorem}

Theorem~\ref{thm:linearThm} establishes some
properties for the solutions of a linear system.

\begin{theorem}
\label{thm:linearThm}
Let \system be a linear system of
$n$ variables and of
$r$ linear equalities and/or 
inequalities with integer coefficients each of size at most 
$l$.
Assume that the vector \mbox{$\myvec{x} = x_1, \ldots, x_n$}
is a solution of \system such that for all $i \in 
\{ 1, \ldots, n\}$, $x_i \geq 0$.
Then there is a vector $\myvec{x^*} = x^*_1, \ldots, x^*_n$
with the following properties:
\begin{enumerate}[label=(\arabic*), topsep = 0.3em]
\item
\label{enum:linearThmSol}
$\myvec{x^*}$ is a solution of \system.
\item
\label{enum:linearThmNonNeg}
for all $i \in \{ 1, \ldots, n\}$, $x^*_i \geq 0$.
\item
\label{enum:linearThmRpos}
at most $r$ entries of $\myvec{x^*}$ are positive.
\item
\label{enum:linearThmPres}
for all $i \in \{ 1, \ldots, n\}$, if $x^*_i > 0$ then $x_i > 0$.
\item
\label{enum:linearThmSize}
for all $i$, $x^*_i$ is a non-negative
rational number with size bounded by
\[
2 \cdot \big ( r \cdot l+ r \cdot \log_2 (r) + 1 \big )~.
\]
\end{enumerate}
\end{theorem}

\begin{proof}
In \system we replace the variables
that correspond to the entries of $\myvec{x}$ that are equal
to zero (if any) with zeros. This way we obtain a
new linear system $\system_0$, with $r$ linear equalities
and/or inequalities and $m \leq n$ variables.
$\myvec{x}$
is a solution\footnote{In the proof of Theorem~\ref{thm:linearThm} all vectors
have $n$ entries. The entries of the vectors
are assumed to be in one to one
correspondence with the variables that appear in the
original system \system.

Let $\myvec{y}$ be
a solution of a linear system $\mathcal{T}$.
If $\myvec{y}$ has more entries than the 
variables of $\mathcal{T}$
we imply that
entries of $\myvec{y}$ that correspond to
variables that appear in $\mathcal{T}$
compose a solution of $\mathcal{T}$.} 
of $\system_0$. 
It also holds that
any solution of $\system_0$
is a solution\footnote{Assume that system
$\mathcal{T}$ has less variables than system
$\mathcal{T}'$. When we say that any solution
of $\mathcal{T}$ is a solution of $\mathcal{T}'$ we imply
that the missing variables are set to $0$.} of \system.

Assume that the system $\system_0$ contains an inequality of
the form 
\begin{equation}
\label{eq:ineqEx}
b_1 \cdot y_{1} + \ldots + b_m y_{m} ~\op~ c
\end{equation}
for $\op \in \{ < , \leq , \geq, >\}$ where
$y_{1}, \ldots , y_{m}$ are variables of
\system
and $b_1 , \ldots ,$
$b_m, c$ are constants that appear in
\system.
$\myvec{x}$ is a solution of \eqref{eq:ineqEx}.
We replace the inequality
\eqref{eq:ineqEx} in $\system_0$ with the following
equality:
\[
b_1 \cdot y_1 + \ldots + b_m y_{m} = b_1 \cdot 
x_1 + \ldots + b_l \cdot x_m
\]
We repeat this procedure for every inequality of $\system_0$.
This way we obtain a system of linear equalities which
we call $\system_1$. It is easy to see that $\myvec{x}$ is
a solution of $\system_1$ and
that any solution of $\system_1$ is also a
solution of $\system_0$ and thus of \system.

Now we will transform $\system_1$ to another linear
system by applying the following algorithm.\\
\textbf{Algorithm:}\\
We set $i = 1$, $e_i = r$, $v_i = m$, $\myvec{x^i} = \myvec{x}$
and we execute the following steps:
\begin{enumerate}[label = (\roman*)]
\item
\label{enum:linTransEqVarCheck}
If $e_i = v_i$ then go to step
\ref{enum:linTransEqNonZeroDetVarCheck}.
Otherwise go to step \ref{enum:linTransIneqVarCheck}.
\item
\label{enum:linTransEqNonZeroDetVarCheck}
If the determinant of $\system_i$ is non-zero then stop.
Otherwise go to step \ref{enum:linTransDepend}.
\item
\label{enum:linTransIneqVarCheck}
If $e_i < v_i$ then go to step \ref{enum:linTransEqThm},
else go to step \ref{enum:linTransDepend}.
\item
\label{enum:linTransEqThm}
We know that the vector $\myvec{x^i}$ is a non-negative
solution for the system $\system_i$.
From Theorem~\ref{thm:eqThm} we obtain a solution $\myvec{x^{i+1}}$
for the system
$\system_i$ which has at most $e_i$ entries positive.
In $\system_i$ we replace the variables that
correspond to zero entries of the solution $\myvec{x^{i+1}}$
with zeros. We obtain a new system
which we call $\system_{i+1}$ with $e_{i+1}$ equalities
and $v_{i+1}$ variables. 
$\myvec{x^{i+1}}$ is a solution of $\system_{i+1}$
and any solution of $\system_{i+1}$ is a solution
of $\system_i$.
We set $i := i+1$ and
we go to step \ref{enum:linTransEqVarCheck}.
\item
\label{enum:linTransDepend}
From any set of equalities that are linearly dependent we
keep only one equation. We obtain a new system
which we call $\system_{i+1}$ with $e_{i+1}$ equalities
and $v_{i+1} := v_i$ variables. We set $i := i+1$ and
$\myvec{x^{i+1}} := \myvec{x^i}$.
We go to step \ref{enum:linTransEqVarCheck}.
\end{enumerate}
Let $I$ be the final value of $i$ after the execution of
the algorithm. Since the only way for our algorithm to
terminate is through step \ref{enum:linTransEqNonZeroDetVarCheck}
it holds that system $\system_I$
is an $e_I \times e_I$ system of linear equalities with
non-zero determinant (for $e_I \leq r$).
System $\system_I$ is obtained from system
$\system_1$ by replacing some variables that correspond
to zero entries of the solution with zeros. So
any solution of $\system_I$ is also a solution of system
$\system_1$ and thus a solution of
\system. From the algorithm we have that
$\myvec{x^I}$ is a solution of
$\system_I$. Since $\system_I$ has a non-zero
determinant
Cramer's rule can be applied. Hence the vector $\myvec{x^I}$
is the unique solution of system $\system_I$.
Let $x^I_i$ be an entry of
$\myvec{x^I}$. $x^I_i$ will be equal to the following rational
number
\[
\frac{
\begin{vmatrix}
a_{11} & \ldots ~a_{1e_I}\\
& \ddots \\
a_{e_I1} & \ldots~ a_{e_Ie_I}
\end{vmatrix}
}{
\begin{vmatrix}
b_{11} & \ldots ~  b_{1e_I}\\
& \ddots \\
b_{e_I1} & \ldots ~ b_{e_Ie_I}
\end{vmatrix}
}
\]
where all the $a_{ij}$ and $b_{ij}$ are integers that
appear in the original system \system. By properties
of the determinant we know that the numerator and the
denominator of the above rational number
will each be at most equal to
$r! \cdot (2^l - 1) ^ r$. So we have that:
\begin{align*}
|x^I_i| & \leq 2 \cdot \big 
(\log_2 (r! \cdot (2^l - 1) ^ r)  + 1 \big ) 
& \Longrightarrow\\
|x^I_i| & \leq 2 \cdot 
\big ( \log_2 (r^r \cdot 2^{l \cdot r})  + 1 \big ) 
& \Longrightarrow\\
|x^I_i| & \leq 2 \cdot 
\big ( r \cdot \log_2 (r) + l \cdot r  + 1 \big ) 
\end{align*}
As we already mentioned the final
vector $\myvec{x^I}$ is a solution of the original linear
system \system. 
We also have
that all the entries of $\myvec{x^I}$ are non-negative,
at most $r$ of its entries are positive and
the size of each entry of $\myvec{x^I}$ is bounded by
$2 \cdot ( r \cdot \log_2 r + r \cdot l + 1)$. Furthermore,
since the variables that correspond to zero entries of
the original vector $\myvec{x}$ were replaced by zeros, we have
that for every $i$, if the i-th entry of $\myvec{x^I}$
is positive then the i-th entry of $\myvec{x}$ is positive too.
So $\myvec{x^I}$ is the requested vector $\myvec{x^*}$.
\end{proof}

The following theorem is an adaptation of
the small model theorem from \cite{fahame90}. Similar
techniques have also been
used in \cite{ograma09} to obtain decidability for
the logic \LPPTwo.

\begin{theorem}[Small Model Property]
\label{thm:smp}
Let \CS be any constant specification 
and let $A \in \LP$.
If $A$ is \PJCSMeas-satisfiable then it is satisfiable
in a \PJCSMeas-model $M = \langle W, H, \mu, * \rangle$ such 
that:
\begin{enumerate}[label=(\arabic*)]
\item
\label{enum:smpWorlds}
$|W| \leq |A|$
\item
\label{enum:smpAlgebra}
$H = \powerset (W)$
\item
\label{enum:smpFinAddw}
For every $w \in W$,
$\mu (\{ w \})$ is a rational number with size at most
\[
2 \cdot \big ( |A| \cdot || A || + |A| \cdot \log_2 (|A|)
+ 1
\big ) 
\]
\item
\label{enum:smpFinAddV}
For every $V \in H$
\[
\mu (V) = \sum_{w \in V} \mu(\{ w\})
\]
\item
\label{enum:smpAtoms}
For every atom of $A$, $a$, there exists
at most one $w \in W$ such that $*_w \Vdash a$.
\end{enumerate}
\end{theorem}

\begin{proof}
Let \CS be any constant specification and let
$A \in \LP$.
Let $a_1 , \ldots , a_n$ be all the atoms of $A$.
By propositional reasoning (in the logic \PJCS)
we can prove that:
\[
\PJCS \vdash A \leftrightarrow \bigvee_{i=1}^{K} \bigwedge_{j=1}^{l_i} P_{\op_{ij} s_{ij}} 
(\beta^{ij})
\]
where all the $P_{\op_{ij} s_{ij}} (\beta^{ij})$
appear in $A$ and $\op_{ij} \in \{ \geq, <\}$.

By using propositional reasoning again (but this time in
the logic \JCS) we can prove that each $\beta^{ij}$ is
equivalent to a disjunction of some atoms of $A$.
So, by using Lemma~\ref{lem:prop} we have that:
\[
\PJCS \vdash A \leftrightarrow \bigvee_{i=1}^{K} 
\bigwedge_{j=1}^{l_i} P_{\op_{ij} s_{ij}} (\alpha^{ij})
\]
where each $\alpha^{ij}$ is a disjunction of
some atoms of $A$. By Theorem \ref{thm:PJsoundCompl} we
have that for any $M \in \PJCSMeas$:
\begin{equation}
\label{eq:satEquiv}
M \models A \Longleftrightarrow M \models
\bigvee_{i=1}^{K} 
\bigwedge_{j=1}^{l_i} P_{\op_{ij} s_{ij}} (\alpha^{ij})
\end{equation}
Assume that
$A$ is satisfiable. By \eqref{eq:satEquiv} there must exist some 
$i$ such that
\[
\bigwedge_{j=1}^{l_i} P_{\op_{ij} s_{ij}} ( \alpha^{ij} )
\]
is satisfiable.
Let $M' = \langle W', H', \mu', *' \rangle$ be a model
such that:
\begin{equation}
\label{eq:SMPdisjSat}
M' \models
\bigwedge_{j=1}^{l_i} P_{\op_{ij} s_{ij}} ( \alpha^{ij} )
\end{equation}
For every $k \in \{1, \ldots, n\}$ we define:
\begin{equation}
\label{eq:XeqMu}
x_k = \mu'([a_k]_{M'})
\end{equation}
In every world of $M'$ some atom of $A$ must hold.
Thus, we have:
\[
W' = \bigcup^n_{k=1} [a_k]_{M'}
\]
And since $\mu'(W') = 1$ we get:
\begin{equation}
\label{eq:SMPmuOne}
\mu' \Big( \bigcup^n_{k=1} [a_k]_{M'} \Big) = 1
\end{equation}
The $a_k$'s are atoms of the same formula, so we have:
\begin{equation}
\label{eq:SMPdisj}
k \neq k' \Longrightarrow
[a_k]_{M'} \cap [a_{k'}]_{M'} = \emptyset
\end{equation}
By \eqref{eq:SMPmuOne}, \eqref{eq:SMPdisj} and the fact that
$\mu'$ is a finitely additive measure we get:
\[
\sum^n_{k=1} \mu'([a_k]_{M'}) = 1
\]
and by \eqref{eq:XeqMu}:
\begin{equation}
\label{eq:SMPsumOne}
\sum^n_{k=1} x_k= 1
\end{equation}
Let $j \in \{ 1, \ldots , l_i \}$. From
\eqref{eq:SMPdisjSat} we get:
\begin{equation*}
M' \models P_{\op_{ij} s_{ij} } \big (\alpha^{ij} \big ).
\end{equation*}
This implies that
$\mu' ( [\alpha^{ij}]_{M'})\ \op_{ij} \ s_{ij}$,
i.e.
\begin{equation*}
\mu' \Bigg ( \Bigg [\bigvee_{a_k \in \alpha^{ij}} a_k 
\Bigg ]_{M'} \Bigg )\  
\op_{ij} \ s_{ij}
\end{equation*}
which implies that
\begin{equation*}
\mu' \Bigg ( \bigcup_{ a_k \in \alpha^{ij}}
[ a_k ]_{M'} \Bigg )\ \op_{ij} \ s_{ij}
\end{equation*}
By \eqref{eq:SMPdisj} and the additivity
of $\mu'$ we have that:
\begin{equation*}
\sum_{a_k \in \alpha^{ij}} \mu' ([a_k]_{M'})\  \op_{ij} \ 
s_{ij}
\end{equation*}
and by \eqref{eq:XeqMu}:
\[
\sum_{a_k \in \alpha^{ij}} x_k \  \op_{ij} \ s_{ij}~.
\]
So we have that
\begin{equation}
\label{eq:sumOpS}
\text{for every } j \in \{ 1, \ldots , l_i\},
\sum_{a_k \in \alpha^{ij}} x_k \  \op_{ij} \ s_{ij}
\end{equation}

Let \system be the following linear system:
\begin{align}
\nonumber
\sum^n_{k=1}& z_k = 1\\
\nonumber
\sum_{a_k \in \alpha^{i1}} z_k \  & \op_{i1} \ s_{i1}\\
\nonumber & \vdots\\
\nonumber
\sum_{a_k \in \alpha^{il_i}} z_k \  & \op_{il_i} \ s_{il_i}
\end{align}
where the variables of the system are $z_1, \ldots, z_n$.
We have the following:
\begin{enumerate}[label=(\roman*), topsep = 0em]
\item
\label{enum:condSol}
By \eqref{eq:SMPsumOne} and \eqref{eq:sumOpS}
the vector $\myvec{x} = x_1, \ldots, x_n$ is a solution
of \system.
\item
From \eqref{eq:XeqMu} every $x_k$ is non-negative. 
\item
Every $s_{ij}$ is a rational number with size at most
$||A||$.
\item
\label{enum:condNumEq}
System \system has at most $|A|$ equalities
and inequalities.
\end{enumerate}
From \ref{enum:condSol}-\ref{enum:condNumEq} and 
Theorem~\ref{thm:linearThm} we have that there exists
a vector $\myvec{y} = y_1, \ldots, y_n$ such that:
\begin{enumerate}[label=(\Roman*), topsep = 0em]
\item
\label{enum:ySolS}
$\myvec{y}$ is a solution of \system.
\item
\label{enum:sizeOfy}
every $y_i$ is a non-negative
rational number with size at most
\[
2 \cdot \big( |A| \cdot ||A||+ |A| \cdot \log_2 (|A|) + 1
\big )~.
\]
\item
at most $|A|$ entries of $\myvec{y}$ are positive.
\item
\label{enum:posInvar}
for all $i$, if $y_i > 0$ then $x_i > 0$.
\end{enumerate}
Assume that
$y_1, \ldots, y_N$ are the positive entries of $\myvec{y}$
where 
\begin{equation}
\label{eq:NleqA}
N \leq |A|
\end{equation}
We define the quadruple $M = \langle W, H, \mu, * \rangle$
as follows:
\begin{enumerate}[label=(\alph*), topsep = 0em]
\item
\label{enum:Mworlds}
$W = \{ w_1 , \ldots, w_N\}$, for some $w_1 , \ldots, w_N$.
\item
\label{enum:Malg}
$H = \powerset (W)$.
\item
\label{enum:MfinAdd}
for all $V \in H$:
\[
\mu(V) =  \sum_{w_k \in V} y_k~.
\]
\item
\label{enum:Meval}
Let $i \in \{ 1, \ldots, N\}$.
We define $*_{w_i}$ to be some
\JCS-evaluation that satisfies the
atom $a_i$. Since $y_i$ is positive, by 
\ref{enum:posInvar},
$x_i$ is positive too, i.e. $\mu'([a_i]_{M'})>0$, which
means that $[a_i]_{M'} \neq \emptyset$, i.e.
that the atom $a_i$ is \JCS-satisfiable.
\end{enumerate}
It holds:
\begin{align*}
\mu(W) & = \sum_{w_k \in W} y_k\\
       & = \sum_{k =1}^n y_k \\
       & \stackrel{\ref{enum:ySolS}}{=} 1
\end{align*}

Let $U, V \in H$ such that $U \cap V = \emptyset$.
It hods:
\begin{align*}
\mu(U \cup V) & = \sum_{w_k \in U \cup V} y_k\\
       & = \sum_{w_k \in U} y_k + \sum_{w_k \in V} y_k\\
       & = \mu(U) + \mu(V)
\end{align*}
Thus $\mu$ is a finitely additive measure. By 
Definitions~\ref{def:PJmodel} and \ref{def:measModel} we have 
that $M \in \PJCSMeas$.

We will now prove the following statement:
\begin{equation}
\label{eq:wKaK}
(\forall 1 \leq k \leq n)  \big [ w_k \in [\alpha^{ij}]_M 
\Longleftrightarrow a_k \in \alpha^{ij} \big ]
\end{equation}

Let $k \in \{1, ~\ldots~, ~n\}$.
We prove the two directions of \eqref{eq:wKaK} separately.

$(\Longrightarrow:)$
Assume that $w_k \in [\alpha^{ij}]$. This means that
$*_{w_k} \Vdash \alpha^{ij}$.
Assume that $a_k \notin \alpha^{ij}$. Then, since
$\alpha^{ij}$ is a disjunction of atoms of $A$, there must
exist some $a_{k'} \in \alpha^{ij}$, with $k \neq k'$, such that
$*_{w_k} \Vdash a_{k'}$. 
However, by definition we have that $*_{w_k} \Vdash a_k$.
But this is a contradiction, since
$a_k$ and $a_{k'}$ are different
atoms of the same formula, which means
that they cannot be satisfied by the same \JCS-evaluation.
Hence, $a_k \in \alpha^{ij}$.

$(\Longleftarrow:)$
Assume that $a_k \in \alpha^{ij}$. We know that $*_{w_k} \Vdash
a_k$, which implies that $*_{w_k} \Vdash \alpha^{ij}$, i.e.~$w_k
\in [\alpha^{ij}]_M$.

Hence, \eqref{eq:wKaK} holds.
Now, we will prove the following statement:
\begin{equation}
\label{eq:SMPsat}
\big ( \forall 1 \leq j \leq l_i \big ) \big [
M \models P_{\op_{ij} s_{ij}} \alpha^{ij} \big ]
\end{equation}

Let $j \in \{ 1, \ldots, l_i \}$.
It holds
\begin{align*}
M & \models P_{\op_{ij} s_{ij}} (\alpha^{ij})
& \Longleftrightarrow\\
\mu ([& \alpha^{ij}]_M) ~\op_{ij}~ s_{ij}
& \Longleftrightarrow\\
\sum_{w_k \in [\alpha^{ij}]_{M}}& y_k
~\op_{ij}~ s_{ij} & 
\stackrel{\eqref{eq:wKaK}}{\Longleftrightarrow} \\
\sum_{a_k \in \alpha^{ij}}& y_k ~\op_{ij}~ s_{ij}
\end{align*}
The last statement holds
because of \ref{enum:ySolS}. Thus,
\eqref{eq:SMPsat} holds.

By \eqref{eq:SMPsat} we have that $M \models 
\bigwedge^{l_i}_{j=1} P_{\op_{ij} s_{ij}} (\alpha^{ij})$, which 
implies that
\[
M \models 
\bigvee_{i=1}^{K}\bigwedge^{l_i}_{j=1} P_{\op_{ij} s_{ij}}
(\alpha^{ij}),
\]
which, by \eqref{eq:satEquiv}, implies that $M \models A$.

Let $w_k \in W$. It holds:
\begin{equation}
\label{eq:yIsMuW}
\mu( \{ w_k \}) = \sum_{w_i \in \{ w_k\}} y_i = y_k
\end{equation}

Now we will show that conditions
\ref{enum:smpWorlds}--\ref{enum:smpAtoms} in
the theorem's statement hold. 
\begin{itemize}
\item
Condition
\ref{enum:smpWorlds} holds because of
\ref{enum:Mworlds} and \eqref{eq:NleqA}.
\item
Condition \ref{enum:smpAlgebra} holds because
of \ref{enum:Malg}.
\item
Condition
\ref{enum:smpFinAddw} holds because of
\eqref{eq:yIsMuW} and \ref{enum:sizeOfy}.
\item
For every $V \in H$, because of \eqref{eq:yIsMuW},
we have:
\[
\mu( V) = \sum_{w_k \in V} y_k
= \sum_{w_k \in V} \mu(\{ w_k\})
\]
Hence condition \ref{enum:smpFinAddV} holds.
\item
By \ref{enum:Meval} every world of $M$ satisfies a
unique atom of $\alpha$. Thus condition \ref{enum:smpAtoms}
holds.
\end{itemize}
So $M$ is the model in question.
\end{proof}

\section{Complexity}

\label{sec:compl}

Lemmata~\ref{lem:atomEval} and~\ref{lem:atomSat} can
be proved by straightforward induction on the
complexity of the formula.
Lemma~\ref{lem:atomEval} tells us that if two \JCS-evaluations agree on some atom of a justification formula
then they
agree on the formula itself.

\begin{lemma}
\label{lem:atomEval}
Let \CS be any constant specification.
Let $\alpha \in \LJ$ and let $a$ be an atom of $\alpha$.
Let $*_1, *_2$ be two \JCS-evaluations and assume that
\[
*_1 \Vdash a \Longleftrightarrow *_2 \Vdash a~.
\] 
Then we have:
\[
*_1 \Vdash \alpha \Longleftrightarrow *_2 \Vdash \alpha~.
\]
\end{lemma}

\begin{lemma}
\label{lem:atomSat}
Let $\alpha \in \LJ$ and let $a$ be an atom of $\alpha$.
Let $*$ be a \JCS-evaluation and assume that
$* \Vdash a$. The decision problem
\begin{center}
does $*$ satisfy $\alpha$?
\end{center}
belongs to the complexity class $P$.
\end{lemma}

Kuznets~\cite{Kuz00CSLnonote} presented an
algorithm for the
\JCS-satisfiability problem
for a total constant specification \CS. 
Kuznets' algorithm
is divided in two parts: the saturation algorithm and the
completion algorithm. Let $\alpha \in \LJ$
be the formula that is tested for satisfiability.
\begin{itemize}
\item 
The saturation algorithm produces
a set of requirements that should be satisfied by any
\JCS-evaluation that satisfies $\alpha$.
The saturation algorithm operates in $NP$-time\footnote{A
reader unfamiliar with notions of computational complexity
theory may consult a textbook on the field, like
\cite{papad94}.}.
\item
The completion algorithm determines whether a
\JCS-evaluation that satisfies $\alpha$ exists or not.
The completion algorithm operates in $coNP$-time.
\end{itemize}
If the saturation and the completion algorithm are taken
together, then we obtain a $\Sigma^p_2$-algorithm for the
\JCS-satisfiability problem (for a total \CS).
The completion
algorithm (adjusted to our notation) is stated in
Theorem~\ref{thm:JCSSP2}.

\begin{theorem}
\label{thm:JCSSP2}
Let \CS be a total
constant specification.
Let $a$ be an atom of some \LJ-formula.
The decision problem
\begin{center}
is $a$ \JCS-satisfiable?
\end{center}
belongs to the complexity class $coNP$.
\end{theorem}

Now we are ready to prove the upper bound for
the complexity of the \PJCSMeas-satisfiability
problem.

\begin{theorem}
\label{thm:PJCSSP2}
Let \CS be a total constant specification.
The \PJCSMeas-satisfiability
problem belongs to the complexity
class $\Sigma^p_2$.
\end{theorem}

\begin{proof}
First we will describe an algorithm that decides the
problem in question and we will explain its correctness.
Then we will evaluate the complexity of the algorithm.

\noindent\textbf{Algorithm:}\\
Let $A \in \LP$.
It suffices to guess a small model $M = \langle W, H , \mu, 
*\rangle$ that satisfies $A$ and
also satisfies the 
conditions~\ref{enum:smpWorlds}--\ref{enum:smpAtoms}
that appear in the statement of Theorem~\ref{thm:smp}.
We guess $M$ as follows: we guess $n$ atoms of $A$, call them
$a_1, \ldots , a_n$, and we also choose
$n$ worlds, $w_1, \ldots, w_n$, for $n \leq |A|$.
Using Theorem~\ref{thm:JCSSP2} we verify that for each
$i \in \{ 1, \ldots, n\}$ there exists a
\JCS-evaluation $*_i$ such that $*_i \Vdash a_i$. 
We define $W = \{ w_1 , \ldots, w_n\}$.
For
every $i \in \{ 1, \ldots, n\}$ we set $*_{w_i} = *_i$.
Since we are only interested in the
satisfiability of justification formulas that appear in
$A$, by Lemma~\ref{lem:atomEval}, the choice of the $*_{w_i}$
is not important (as long as $*_{w_i}$ satisfies
$a_i$).

We assign to every $\mu(\{w_i\})$ a rational number
with size at most:
\[
2 \cdot \big ( |A| \cdot ||A|| + |A| \cdot \log_2 
(|A|) + 1 \big )~.
\]
We set $H = \powerset(W)$. For every $V \in H$ we
set:
\[
\mu(V) = \sum_{w_i \in V} \mu (\{w_i\})~.
\]
It is then straightforward to see that the 
conditions~\ref{enum:smpWorlds}--\ref{enum:smpAtoms}
that appear in the statement of Theorem~\ref{thm:smp} hold.

Now we have to verify that our guess is correct, i.e. that
$M \models A$.
Assume that $P_{\geq s} \alpha$ appears in $A$.
In order to see whether $P_{\geq s} \alpha$ holds
we need to calculate the measure of the set $[\alpha]_M$
in the model $M$. The set $[\alpha]_M$ will contain
every $w_i \in W$ such that $*_{w_i} \Vdash \alpha$.
Since $*_{w_i}$ satisfies an atom of $A$ it also
satisfies an atom of $\alpha$. So, by 
Lemma~\ref{lem:atomSat}, we can check whether $*_{w_i}$
satisfies $\alpha$ in polynomial time.
If $\sum_{w_i \in [\alpha]_M} 
\mu(\{w_i\}) \geq s$ then we replace $P_{\geq s} \alpha$
in $A$
with the truth value \true, otherwise with
the truth value \false.
We repeat the above procedure for every formula
of the form $P_{\geq s} \alpha$ that 
appears in $A$. At the end we have a formula that is
constructed only from the connectives $\lnot$, $\land$
and the truth constants \true and \false. Using
a truth table we can verify in polynomial
time that the formula is true.
This, of course implies that $M \models A$.

\noindent\textbf{Complexity Evaluation:}\\
All the objects that are guessed in our
algorithm have size that is polynomial on $A$.
Also the verification phase of our algorithm can
be made in polynomial time. Furthermore the application
of Theorem~\ref{thm:JCSSP2} is possible with an
$NP$-oracle (an $NP$-oracle can obviously decide
$coNP$ problems too). Thus our algorithm is
an $NP^{NP}$ algorithm and since  
$\Sigma^p_2 = NP^{NP}$ the claim of the Theorem
follows.
\end{proof}

\section{Final Remarks and Conclusion}

\label{sec:concl}

As a continuation of \cite{komaogst} and
\cite{koogst} we showed that results
for justification logic and probabilistic logic can be nicely
combined. 
Recall that the probabilistic justification logic \PJ is obtained by 
adding probability operators to 
the justification logic \J. In~\cite{Kuz00CSLnonote} it
was proved that under some assumptions on the 
constant specification
the complexity of the satisfiability problem for the logic
\J belongs to the class $\Sigma^p_2$. By 
Theorem~\ref{thm:PJCSSP2}
we have that, under the same assumptions on
the constant specification,
the complexity of the satisfiability problem
for the logic \PJ remains in the same complexity
class. Hence, the probabilistic operators do not
increase the complexity of the satisfiability problem,
although they increase the expressiveness of the
language.

As it is pointed out in \cite{Kuz08PhD}, Theorem \ref{thm:JCSSP2}
holds for a decidable almost schematic constant
specification. Theorem \ref{thm:PJCSSP2} uses
Theorem \ref{thm:JCSSP2} as an oracle.
So, obviously Theorem \ref{thm:PJCSSP2}
holds for a decidable almost schematic constant
specification too.

The upper complexity bound we established
is tight. By a
result from \cite{Mil07APAL} which was later strengthened
in \cite{BusKuz12APALnonote} and \cite{ach15} we
have that for a decidable, schematic and
axiomatically appropriate constant specification \CS the
\JCS-satisfiability problem is $\Sigma^p_2$-hard. For any
$\alpha \in \LJ$ it is not difficult to prove that:
\[
\alpha \text{ is \JCS-satisfiable } \Longleftrightarrow
P_{\geq 1} \alpha \text{ is \PJCSMeas-satisfiable}
\]
Hence, the \JCS-satisfiability problem can be reduced to
the \PJCSMeas-satisfiability problem, which
implies that the \PJCSMeas-satisfiability
problem is $\Sigma^p_2$-hard too.
Thus the \JCS-satisfiabilty problem as well as
the \PJCSMeas-satisfiability problem
are $\Sigma^p_2$-complete.

Observe that by Theorem~\ref{thm:PJsoundCompl} and
our previous remarks we
have that, for a decidable schematic and axiomatically
apropriate constant specification,
the derivability problem for the logic \PJCS is
$\Pi^p_2$-complete.

In \cite{koogst} the probabilistic justification logic 
\PPJ is defined. \PPJ is a natural extension of \PJ
that supports iterations of the probability operator as
well as justifications over probabilities.
An interesting open problem related to the present work
is to determine complexity bounds for \PPJ.

\vspace*{1em}

\noindent{\bf Funding:}\\
The author is supported by the SNSF project 153169, 
\emph{Structural Proof Theory and the Logic of Proofs}.

\vspace*{1em}

\noindent\textbf{Acknowledgements:}\\
The author is grateful to 
Antonis Achilleos, Thomas Studer and the anonymous referees for 
valuable comments and remarks that
helped him improve the quality of the
paper substantially.

\end{document}